\documentclass{article}
\title{On Adversary Robust Consensus protocols through joint-agent interactions}
\author{David Angeli and Sabato Manfredi}
\usepackage{geometry}
\usepackage[dvips]{epsfig}    
\usepackage{amsmath,accents}%
\usepackage{amsfonts}%
\usepackage{amssymb,subfigure}%
\usepackage{graphicx}
\usepackage{amsfonts}
\usepackage{amssymb}
\usepackage{amsmath,epsfig,graphicx,array,setspace,subfigure,subfloat}
\usepackage{amsfonts}
\usepackage{amssymb}
\usepackage{color}
\usepackage{amsmath,epsfig,graphicx,array,setspace,subfigure,subfloat}
\usepackage{amsfonts}
\usepackage{amssymb}
\usepackage{amsmath}%
\usepackage{amsfonts}%
\usepackage{amssymb,subfigure}%
\usepackage{graphicx}
\usepackage{amsfonts}
\usepackage{amssymb}
\usepackage{color}
\usepackage{algorithm}
\usepackage{algpseudocode}
\def\undertilde#1{\mathord{\vtop{\ialign{##\crcr
$\hfil\displaystyle{#1}\hfil$\crcr\noalign{\kern1.5pt\nointerlineskip}
$\hfil\tilde{}\hfil$\crcr\noalign{\kern1.5pt}}}}}

\newtheorem{lemma}{Lemma}

\newtheorem{definition}{Definition}

\newtheorem{theo}{Theorem}
\newenvironment{proof}[1][Proof]{\begin{trivlist}
\item[\hskip \labelsep {\bfseries #1}]}{\end{trivlist}}
\newcommand{\beq}{\begin{equation}}
\newcommand{\eeq}{\end{equation}}

\begin{document}
\maketitle
\begin{abstract}
A generalized family of Adversary Robust Consensus protocols is proposed and analyzed. These are distributed algorithms for multi-agents systems seeking to agree on a common value of a shared variable, even in the presence of faulty or malicious
agents which are updating their local state according to the protocol rules.
In particular, we adopt monotone joint-agent interactions, a very general mechanism for processing locally available information and allowing cross-comparisons between state-values of multiple agents simultaneously. 
The salient features of the proposed class of algorithms are abstracted as a Petri Net and convergence criteria for the resulting time evolutions formulated by employing structural invariants of the net. 
\end{abstract}
\section{Introduction and motivations}
Algorithms for consensus were introduced in \cite{tsitsi} a few decades ago, in the context of distributed optimization, a topic which remains of great interest still today, \cite{boyd}. \\
The role played by propagation of information in achieving consensus among interacting agents, was first highlighted in the seminal paper \cite{Moreau05}.
Therein, authors formulated tight and explicit graph-theoretical requirements for asymptotic consensus in time-varying linear update protocols, by abstracting the network of agents' interactions and its underlying dynamics as a graph. 
This sparked a considerable interest of the scientific community in advancing and applying consensus protocols for multi-agent systems (see i.e. \cite{murray_survey, Hendrickx12, Martin16} and references therein).
Subsequent developments in the theory of nonlinear consensus protocols have formalized and clarified the role of information spread along the graph of agents' interaction for more general situations, including second and higher order agents dynamics, or agents' states evolving on manifolds \cite{manifold} or nonlinear interactions \cite{Maggiore07, Manfredi_TAC}.
More recently, graph theoretical criteria have been similarly developed to encompass asymmetric confidence, as in the case of unilateral interactions \cite{ourpaper}, or joint-agent interactions, \cite{angelimanfredijoint}.\\

The latter, in particular, account for situations where individual agents impose ``filtering thresholds'' upon neighbours' influences by cross-validating their opinions through mutual comparisons that only allow for consistent infuences (either from above or below) of  two or more neighbouring agents to be enacted upon.  
In this regard, unlike the majority of existing consensus protocols that implicitly assume `additive' dynamics and exhibit variation rates as a disjunctive combination (sum) of neighbors influences, joint-agent interactions allow the formulation of conjunctive 
influences, and, respectively, of their additive combination.  \\

It is worth stressing that complex contemporary social and engineering systems  often need to deal with selfish or malicious users, node faults and attacks (\cite{Survey1}- \cite{Survey5}). In this respect the evaluation of individual and group reputation play a focal role for the safety of such systems. In the last years different (centralised and distributed) algorithms have been proposed to deal with online reputation estimation of both individual (\cite{back1}-\cite{back6}) and group (clustering all users according to their rating similarities - \cite{Group_rank1,Group_rank2}) to providing incentives to users acting responsibly and cooperatively.  

Within this line of investigation,  the problem of Adversary Robust Consensus Protocols (ARC-P) was formulated in \cite{leblanc1}, following earlier seminal results in \cite{pease}. Therein, Leblanc and coworkers propose and analyze a discrete time protocol which allows $n$ cooperating agents to converge towards a consensus state, within a complete all-to-all network, even when
a subset of agents (of cardinality up to $\lfloor n/2 \rfloor$) is \emph{malicious} or \emph{faulty}, namely it evolves in a completely arbitrary way, with the sole constraint of broadcasting its own state to all remaining agents.
The proposed protocol simply orders state values in ascending (or descending order) and removes  $F$ top and lowest values from the ordered list, where $F$ is an apriori fixed bound to the number of malicious agents. Then, the average among the remaining values is computed and a standard linear consensus update equation is applied. \\

Subsequent analysis has been devoted in \cite{leblanc3} to the important topic of relaxing the all-to-all topology requirement and investigating sufficient conditions for Adversary Robust Consensus on the basis of local information only, or in the presence of so called Byzantine agents \cite{leblanc2}, who may, either intentionally or due to faulty conditions, communicate different state values to different neighbours. A related line of investigation assumes the presence of trusted nodes, \cite{abbas}.\\

The interconnection topology is interpreted, in such context, as a specific type of switching (time-varying) linear consensus, arising through the application of the so called \emph{sorting} function, its composition with the \emph{reducing} function (responsible for discarding highest and lowest values) and finally by averaging the entries of the vector obtained. It turns out, however, that similar types of agents interactions can also be recast within the framework of joint-agent interactions.


The formalism of joint-agent interactions is, in this respect, even more flexible and may, for instance, allow to partition neighbours of every agent in several subgroups, to be suitably \emph{sorted, reduced} and \emph{averaged} while adding (possibly with different weights) the influences resulting from distinct subgroups as a final step.
This type of rules for processing local information results in consensus protocols which allow different levels of trust attributed to different set of neighbors and, generally speaking, break the simmetry implicit in the use of a single sorting and reducing function.  \\

The extended class of intrinsically nonlinear consensus protocols afforded by the use of joint-agent interactions can be conveniently described and characterized, from a topological point of view, as \emph{bipartite graphs}, and more specifically Petri Nets.
It turns out that structural notions, developed in the context of Petri Nets to ascertain their liveness as Discrete Event Systems, play a crucial role in characterizing the ability of a network of agents to reach consensus regardless of initial conditions
\cite{angelimanfredijoint}. \\

The specific details of the conditions needed for this to happen will be illustrated in a subsequent Section. Nevertheless, it is intuitive that if, on one hand, application of conjunctive filtering conditions among neighbouring agents limits the spread of information across the network (and therefore, if not done carefully may prevent consensus from happening at all), on the other,  it only allows ``trustworthy'' information to be propagated, and therefore may result (if carefully deployed) in Adversary Robust Consensus protocols.
In this paper we address the issue of when a network of agents, with arbitrary (and possibly asymmetric) interconnection topology (allowing for instance differentiated trust levels among neighboors) exhibits the ability to reach consensus despite a subset of its agents being either \emph{faulty} or \emph{malicious}, viz. able to influence other nodes according to their individual
state-value but, in fact, upgrading their position in a completely arbitrary fashion. 


Just to illustrate the potential of the approach, we present below simulations referring to an all-to-all network of $5$ agents, involving linear interactions, or a similar network entailing  joint-agent interactions.
Our theory allows to prove that, in the latter network, robustness can be achieved allowing any set of $2$ out of $5$ agents to be malicious or faulty, and still guaranteeing the remaining healthy agents will retain the ability to reach exact consensus.
In particular, we simulate the following linear network:
\begin{equation}
\label{linlike}
\dot{x}_i = \sum_{j \neq i} a_{ij} (x_j - x_i), \qquad i,j \in \{1 \ldots 5 \},  
\end{equation}
for some $a_{ij}>0$ and compare it with the following nonlinear all to all network:
\begin{equation}
\label{joint}
\dot{x}_i = \sum_{J  \subset \{1 \ldots 5 \} \backslash \{ i \}: |J|=3}
f_{J \rightarrow i} (x)
\end{equation}
where the function $f_{J \rightarrow i}: \mathbb{R}^n \rightarrow \mathbb{R}$ is defined below:
\[ f_{J \rightarrow i} (x) = \max_{j \in J} \; \min \{ x_j - x_i,0 \}   + \min_{j \in J} \; \max \{x_j - x_i ,0 \}. \]
We pick the initial condition $[35,10,5,15,20]'$ and run the two consensus protocols assuming that agents $4$ and $5$ are faulty and follow the apriori fixed time evolutions:
\[ x_4 (t) = 15 + \frac{\cos\!\left(3\, t\right) - 1}{9} + \frac{t\, \sin\!\left(3\, t\right)}{3} + \frac{t^3}{150} \]
\[
 x_5(t)=20 + \frac{\sin\!\left(2\, t\right)}{4} + \frac{t\, \left(2\, {\sin\!\left(t\right)}^2 - 1\right)}{2}. \]
These agents  effectively act as exogenous disturbance inputs for the remaining $3$ agents, and, from the practical point of view, may be regarded as faulty agents or malicious agents trying to disrupt the consensus.
In the case of equation (\ref{linlike}), as it is expected due to linearity and addivity of interactions, the exogenous disturbances $x_4$ and $x_5$ are able to spread their influence to the remaining agents and effectively prevent the remaining agents to asymptotically reach consensus (see Fig. \ref{lineardisruption}).
\begin{figure}
\centerline{
\includegraphics[width=12cm]{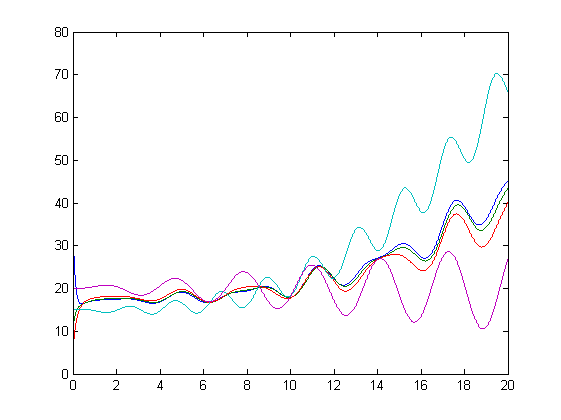}}
\caption{Linear consensus protocol subject to faulty agents $4$ and $5$}
\label{lineardisruption}
\end{figure} 
In the case of equation (\ref{joint}), instead, nonlinear joint-agent interactions allow the three``healthy'' agents, $1,2,3$ to asymptotically reach agreement within the convex-hull of their initial values, regardless of $x_4$ and $x_5$, (see Fig. \ref{jointrobust}).
\begin{figure}
\centerline{
\includegraphics[width=12cm]{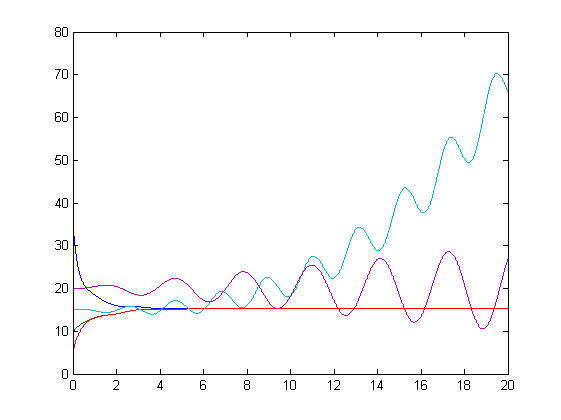}}
\caption{Joint-Agent consensus protocol subject to faulty agents $4$ and $5$}
\label{jointrobust}
\end{figure}
In other words, while the faulty agents may, to a certain extent, affect the final consensus value reached, they are unable to disrupt it.
\section{Problem formulation}
The aim of this note is to derive necessary and sufficient conditions to characterize when networks of agents implementing joint-agent interactions may be able to achieve \emph{robust} consensus  in the presence of possibly \emph{malicious} or \emph{faulty} agents.
In particular, we study networks described by the following class of nonlinear finite-dimensional differential equations:
\begin{equation}
\label{net}
\dot{x} = f(x)
\end{equation}
where $x \in \mathbb{R}^n$ is the state vector, and $f: \mathbb{R}^n \rightarrow \mathbb{R}^n$ is a Lipschitz continous function, describing the update laws of each agent as a function of its own and neighbours' state values.
For convenience we ask that $f_j$ be monotonically non-decreasing with respect to all $x_i$ ($i \neq j$) so that the resulting flow is monotone with respect to initial conditions, once the standard order induced by the positive orthant is adopted.
This assumption, while not essential, can to a certain extent simplify the analysis and the definition of interaction among agents.
Many authors, in recent years, have elaborated conditions under which solutions of (\ref{net}) asymptotically converge towards equilibriums of the following form:
\begin{equation}
\lim_{t \rightarrow + \infty} \varphi(t,x_0) = \bar{x} \textbf{1}
\end{equation}
 for some $\bar{x} \in \mathbb{R}$, where $\textbf{1}$ is the vector of all ones in $\mathbb{R}^n$. When this occurs for all solutions, and regardless of initial
conditions, we say that system (\ref{net}) achieves \emph{global asymptotic consensus}.  \\ 

In this note, however, we consider a more general situation in which the state vector $x$ is partitioned into two subvectors, $x_H$ and $x_F$, associated to \emph{healthy} and \emph{faulty} agents respectively.  Accordingly, we denote 
$f(x) = [f_H (x)',f_F(x)']'$ and  wish to characterize under what assumptions
solutions of 
\begin{equation}
\label{projectednet}
\dot{x}_H = f_H (x_H, x_F) 
\end{equation}
asymptotically converge to equilibria of the type:
\[ \lim_{t \rightarrow + \infty} x_H(t) = \bar{x} \textbf{1}_H \]
for all initial conditions $x_H(0)$ and all exogenous input signals $x_F(\cdot)$. A formal definition follows.
\begin{definition}
We say that network (\ref{net}) achieves robust consensus in the face of faults in $F \subset \mathcal{N}$, if,
partitioning the state vector according to $F$ and $H:=\mathcal{N} \backslash F$ yields:
\[ \lim_{t \rightarrow + \infty} \varphi_{H} ( t, x_H(0), x_F( \cdot) ) = \bar{x}   \textbf{1}_H \]
for all $x_H(0)$, and all uniformly bounded exogenous input $x_F(\cdot)$, (where $\varphi_H (t, x_H(0), x_F(\cdot))$ denotes the solution
of (\ref{projectednet}) at time $t$ from initial condition $x_H(0)$ and input $x_F(\cdot)$).
\end{definition}

In practice, for a given net, we will be interested in considering several possible combinations of faulty agents (corresponding to several choices of $F$) and, for each one of them, verify conditions for asymptotic convergence towards consensus of the remaining \emph{healthy} agents $H$. \\
\newline
In order to characterize the flow of information needed for achieving such kind of behaviour, we recall the notion of \emph{joint agent interaction}, as proposed in \cite{angelimanfredijoint}. 

\begin{definition}
\label{jointbilateral}
We say that a group of agents $I \subset \mathcal{N}$ jointly influences agent $j \in \mathcal{N} \backslash I$ if for all compact intervals $K \subset \mathbb{R}$ there exists a positive definite function $\rho$, such that, for all $x_I, x_j \in K$ it holds:
\begin{equation}
\label{jointinter}
 \textrm{sign}(x_I - x_j) f_j ( x_j \textbf{1} + (x_I-x_j) e_I ) \geq \rho ( |x_I - x_j | ). 
\end{equation}
We denote this by the following shorthand notation: $I \rightarrow j$.
\end{definition}
Notice that influence from $I$ to $j$, denoted as $I \rightarrow j$, is monotone (in its first argument $I$) with respect to set-inclusion.
In particular, if $j \in \mathcal{N} \backslash \tilde{I}$ we have:
\[ I \rightarrow j \textrm{ and } I \subset \tilde{I}   \; \Rightarrow \tilde{I} \rightarrow j. \] 
For this reason, it is normally enough to consider \emph{minimal} influences alone. We say that $I$ influences $j$ and that this influence is \emph{minimal} if
there is no $\tilde{I} \subsetneq I$ such that $\tilde{I} \rightarrow j$.

\section{Relevant Petri Net background}
Our goal is to derive characterizations of a graph theoretical nature regarding the ability of networks with joint-agent interactions to exhibit robust consensus, in the face of faults or malicious attacks.
We adopt, to this end, the formalism introduced in \cite{angelimanfredijoint}.
In particular, we represent multiagent networks as Petri Nets. These are a type of bipartite graph, used to model Discrete Event Systems, and can be conveniently adopted in the present study. In fact, a rich literature on structural invariants for Petri Nets already exists, including software libraries to compute them as well as complexity analysis of the available algorithms.

An (ordinary) Petri Net is a quadruple $\{ P, T, E_I , E_O \}$, where
$P$ and $T$ are finite sets (with $P \cap T = \emptyset$) referred to as
 \emph{places} and \emph{transitions}, respectively. These are nodes of a directed bipartite graph.
In fact, directed edges are of two types: $E_I \subset T \times P$ connecting transitions to places and $E_O \subset P \times T$ connecting places to transitions. \\

In our context places represent agents while transitions stand for interactions among them.
More closely, to each agent $i \in \mathcal{N}$ there exists a unique associated place $p_i \in P$.
Furthermore, if agents in $J \subset \mathcal{N}$ jointly influence agent $i$, this is denoted as $J \rightarrow i$ and, provided this interaction is minimal, it is represented graphically by
a single transition $t \in T$,  with edges $(p_j, t) \in E_O$ for all $j \in J$ and a single edge $(t,p_i )$ in $E_I$.
Notice that every transition can be assumed to only afford exactly one outgoing edge, unlike in general Petri Nets.
As an example, we show in Fig. \ref{netexamples} the graphical representation of the Petri Nets associated to the list of interactions
\begin{equation}
\label{simplestnet}
 \{ 1 \} \rightarrow 2, \; \{ 1,2 \} \rightarrow 3, 
\end{equation}
and, next to it, for a kind of ring topology with $5$ agents and the following list of minimal joint agent interactions:
\begin{equation}
\label{ring5}
 \{1,2\} \rightarrow 3, \; \{2,3\} \rightarrow 4, \; \{3,4\} \rightarrow 5, \; \{4,5 \} \rightarrow 1, \; \{5, 1 \} \rightarrow 2. 
\end{equation}
\begin{figure}
\centerline{
\includegraphics[height=3.5cm]{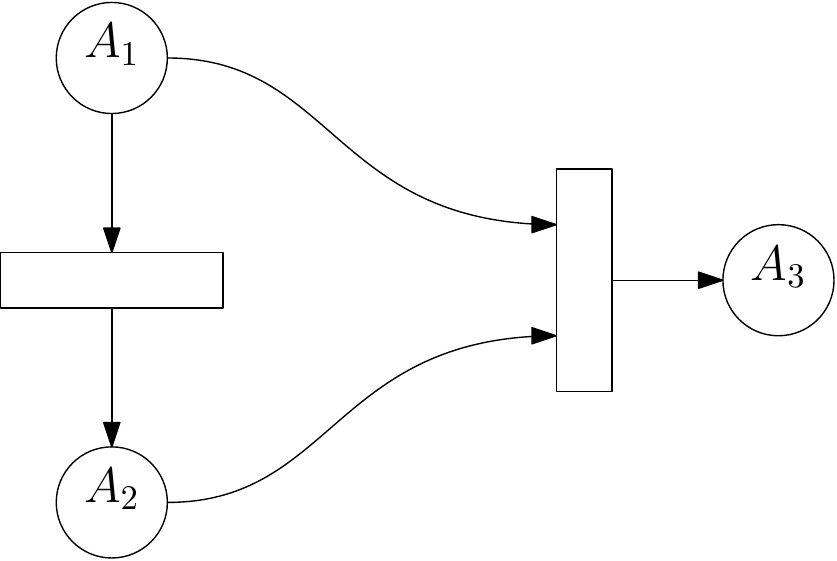} \,
\includegraphics[height=3.5cm]{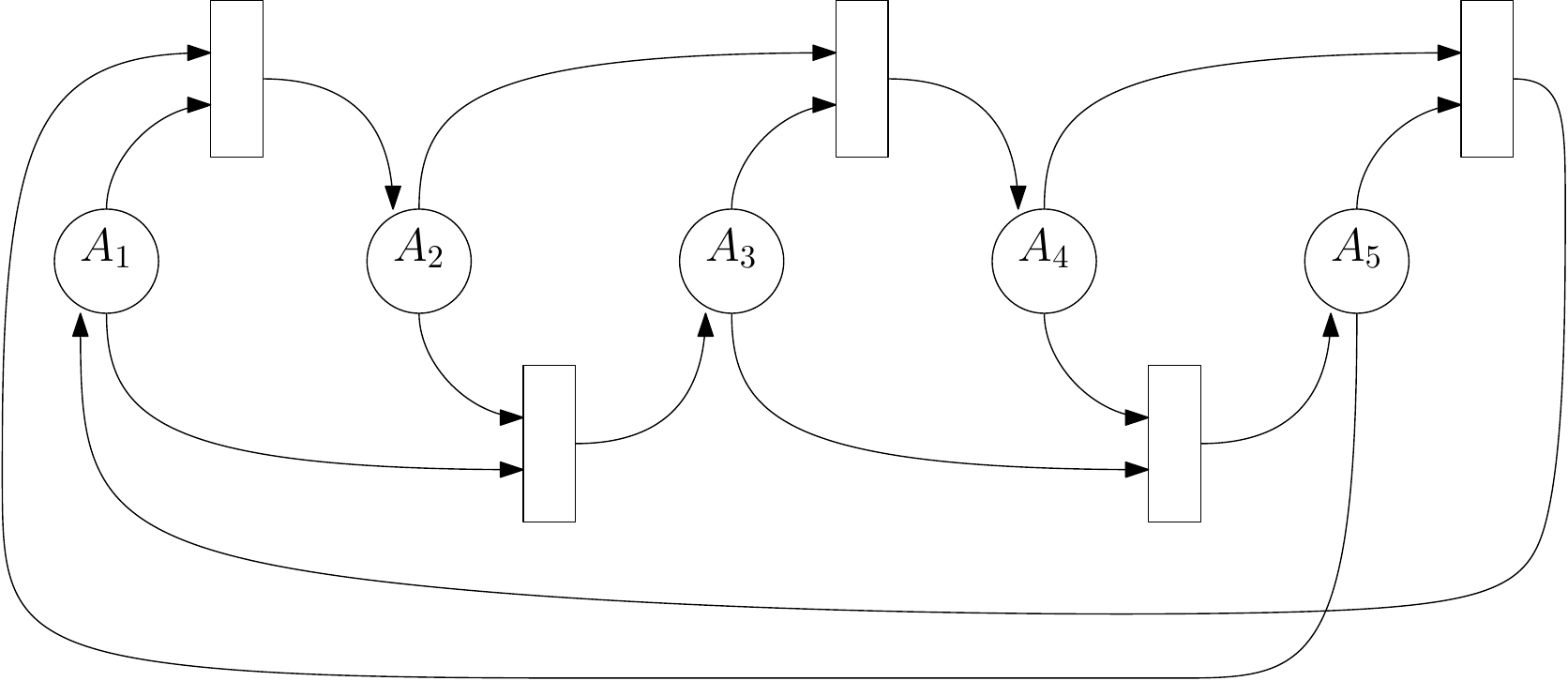} }
\caption{Petri Nets associated to network of interactions (\ref{simplestnet}) and (\ref{ring5}) }
\label{netexamples}
\end{figure}
The next concepts will be crucial in characterizing, from the topological point of view, networks that guarantee asymptotic convergence towards consensus.
The set of \emph{input transitions} for a place $p$, is denoted as
\[ I(p) = \{ t \in T: (t,p) \in E_I \}, \]
and, similarly for a set of places $S \subset P$, its input transitions are:
\[ I(S) = \{ t \in T: \exists \, p \in S: (t,p) \in E_I \}. \]
Simmetrically, output transitions are denoted as:
\[ O(S) = \{ t \in T: \exists \, p \in S: (p,t) \in E_O \}. \]
\begin{definition}
\label{siphondef}
A non-empty set of places $S \subset P$ is called a \emph{siphon}
if $I(S) \subset O(S)$.  A siphon is minimal if no proper subset is also a siphon.
\end{definition}
Informally, in a group of agents that correspond to a  siphon, any influence needs to come (at least in part) from within the group.
In \cite{angelimanfredijoint}, a  characterization of the ability of agents to asymptotically converge towards consensus (regardless of their initial conditions) is provided.
This feature, called \emph{structural consensuability}, is shown to be equivalent to the requirement that any pair of siphons in the associated Petri Net have non-empty intersection.

To address robustness questions, within the same set-up of joint agent interactions, an externsion of the concept of siphon is needed.
The following is, to the best of our knowledge, an original definition:
\begin{definition}
\label{cs}
A non-empty set of places $S \subset P$ is an $F$-controlled siphon, if:
\[ I(S) \subset O(S) \cup O(F). \]
\end{definition}
Notice that Definition \ref{cs} boils down to the standard notion of siphon for $F= \emptyset$.
Union of $F$-controlled siphons is again an $F$-controlled siphon and, in particular, if a set is an $F_1$-controlled siphon it is also an $F_2$-controlled siphon for all $F_2 \supseteq F_1$. We call the set $F$ the \emph{switch} of siphon $S$.
Informally, this terminology is adopted as malicious agents in $F$ may prevent healthy agents in $S$ from increasing (or decreasing) their own state values. This is indeed achievable by malicious agents simply broadcasting values which are either below the minimum or, respectively, above the maximum of all values within the siphon.

The following notions are appropriate to characterize occurrence of robust consensus.
\begin{definition}
\label{robconscond}
We say that a Petri Net fulfills robust consensuability with respect to faults in $F \subset \mathcal{N}$ if $H := \mathcal{N} \backslash F$ is a siphon and for all
pairs of controlled siphons $S_1$,$S_2$ and associated switches $F_1,F_2 \subset F$, we have the following:
\begin{equation}
\label{robcons}
 S_1 \cap S_2 = \emptyset \Rightarrow F_1 \cap F_2 \neq \emptyset. 
\end{equation}
\end{definition}  
It is worth pointing out that robust consensuability, when $F= \emptyset$, boils down to \emph{structural} consensuability as defined in \cite{angelimanfredijoint}. Also, a direct comparison with previously existing conditions for consensuability in networks where all influences are single-agent influences is not possible, as the condition would never be fulfilled. It is in fact use of joint-agent interactions and cross validations that make adversary robust consensus achievable.
On the other hand, we believe that our conditions boil down to those proposed in \cite{leblanc3} when only interactions obtained through a sorting and reducing function are allowed.


\section{Main result and  proofs}
In the following Section we state the main result and clarify the steps of its proof. To this end, in order to allow dynamical properties of a multiagent system to be derived on the basis of structural conditions fulfilled by the associated Petri Net, it is important to establish a closer link between the considered equations and the associated Petri Net. 
For any transition $t \in T$, denote by:
\[ I(t):= \{ p \in P : (p,t) \in E_O \} \]
and by $j(t)$ the unique place such that $(t, j(t) )$ belongs to $E_I$.
In particular, for a given Petri Net $\{P,T,E_I,E_O \}$ we consider non-decreasing
 locally Lipschitz functions $F_i: \mathbb{R}^{|I(p_i)|} \rightarrow \mathbb{R}$, with $F_i (0,0,\ldots,0)=0$, and such that $F_i$ is strictly increasing in each of its arguments in $0$.  These are employed to define networks of equations:
\begin{equation}
\label{specialstructure}
 \dot{x}_i = F_i ( f_{I(t_1) \rightarrow i} (x), f_{I(t_2) \rightarrow i} (x), \ldots, f_{I(t_{|O(p_i)|}) \rightarrow i} (x) ), \qquad O(p_i)=\{ t_1, \ldots t_{|O(p_i)|} \}.
 \end{equation}
A typical example arises when $F_i ( f ) = \sum_{k} \alpha_k f_k$, for some choice of coefficients $\alpha_k>0$.
Equation (\ref{specialstructure}) is, however, more general and allows non-additive agents' infuences.
As an example of a non-additive function $F_i$,  one may consider for instance the map $F_i (f) = \min_{k \in \{1, \ldots, |O(p_i)| \}} f_k + \max_{k \in \{1 \ldots |O(p_i)| \} } f_k$. \\
Composition of the above maps with monotonic increasing functions, such as saturations or (odd) powers are also legitimate choices,
i.e. $F_i ( f ) = \sum_{k} \alpha_k \textrm{sat} (f_k)$ or $F_i ( f ) = ( \sum_{k} \alpha_k f_k )^3$.

We are now ready to state our main result and, later, to discuss the technical steps of its derivation.
\begin{theo}
\label{mr}
Consider a cooperative network of agents as in (\ref{net}) and let $N$ be the Petri Net associated to its set of \emph{minimal} joint agent interactions. 
 Consider a partition of $\mathcal{N}$ into two disjoint subgroups $F, H \subset\mathcal{N}$, which represent the Faulty and the Healthy agents (respectively), along with the projected dynamics, (\ref{projectednet}). Then, robust consensus is achieved among the agents in $H$ provided $N$ fulfills robust consensuability with respect to faulty agents in $F$.  
 \end{theo}
It is worth pointing out that the result assumes faulty agents are following arbitrary continuous evolutions and that these are reliably broadcast to all neighbouring agents. This hypothesis cannot model the situation in which malicious agents intentionally communicate different evolutions to different neighbors. Agents with this ability are usually referred to as Byzantine agents, and Byzantine consensus protocols exhibit robustness to such kind of threats.
Notice that the ability of malicious agents of differentiating the information sent to neighbors may disrupt consensus even when robust consensuability is fulfilled. An example of this situation is later shown in Section \ref{simuex}.  \\

We start the technical discussion by generalizing Proposition 11 in \cite{angelimanfredijoint}.
\begin{lemma}
\label{invlemma}
Let $H$ be a siphon of $N$. Consider a network of equations (\ref{specialstructure}) and let $x_H$ denote the state vector of agents in $H$, along with the corresponding
equations
\begin{equation}
\dot{x}_H (t) = f_H ( x_H(t), x_F(t) )
\end{equation}
as introduced in (\ref{projectednet}). Then, for any $c \in \mathbb{R}$, the sets:
\[ \bar{\mathcal{X}}_c := \{ x_H \in \mathbb{R}^{|H|}: x_H \leq c \textbf{1}_H \}, \]
\[ \underline{ \mathcal{X}}_c := \{ x_H \in \mathbb{R}^{|H|}: x_H \geq c \textbf{1}_H \} \]
are robustly forward invariant for any bounded input signal $x_F (\cdot)$.
\end{lemma}
\begin{proof}
Let $x_F(\cdot)$ take value in the compact set $K$ and $h \in H$ be any agent whose associated state value fulfills $x_h=c$.  To prove invariance of 
$\bar{\mathcal{X}}_c$ we need to show
$f_h (x_H,x_F) \leq 0$ for all $x_F \in K$. This condition, in fact, amounts to $f (x_H,x_F) \in TC_x ( \bar{ \mathcal{X}}_c )$ for all 
$x_H \in \partial   \bar{\mathcal{X}}_c $ and all $x_F \in K$. This, in turn implies forward invariance of $\bar{\mathcal{X}}_c$ by Nagumo's Theorem. 
Let $O(p_h) = \{ t_1, t_2, \ldots, t_{|O(p_h)|} \}$. Since $h$ is a siphon, for all $t_i$ in $O(p_h)$ there exists
$\tilde{h} \in I (t_i) \cap H$. Hence:
\[ f_{ I (t_i) \rightarrow h } (x) \leq f_{ I (t_i) \rightarrow h } (\bar{x}_H \textbf{1}_H, x_F ) = 0. \]
By monotonicity of $F_h$ then:
\[
  \dot{x}_h = F_h ( f_{I(t_1) \rightarrow h} (x), f_{I(t_2) \rightarrow h} (x), \ldots, f_{I(t_{|O(p_h)|}) \rightarrow h} (x) ) \leq F_h (0,0,\ldots,0) = 0. \]
This completes the proof of the Lemma. 
\end{proof}
The rest of the Section is devoted to illustrate the main technical steps of the proof.
\begin{proof} Let $x_F (t)$ be an arbitrary bounded, continuous signal. Assume, in particular, that $x_F(t) \in K$ for some compact set $K \subset \mathbb{R}^{|F|}$.
Pick any initial agent distribution $x_H(0)$ and define the evolution of healthy agents in $H$ according to the equation (\ref{projectednet}).
In particular, we denote the solution $x_H(t) := \varphi_H (t, x_H(0), x_F( \cdot ) )$, for all $t \geq 0$.   
Moreover, we let:
\[ \bar{x}_H  := \max_{h \in H} x_h  \qquad \underline{x}_H  := \min_{h \in H}  x_h. \]
Since $H$ is a siphon, by Lemma \ref{invlemma}, we see that for all $t_2 \geq t_1 \geq 0$:
\[ x_H (t_1) \in \bar{\mathcal{X}}_{\bar{x}_H(t_1)} \Rightarrow x_H(t_2) \in \bar{\mathcal{X}}_{\bar{x}_H(t_1)}. \]
In particular then, $\bar{x}_H(t_1) \geq \bar{x}_H (t_2)$, viz. $\bar{x}_H$ is monotonically non-increasing. As expected, a symmetric argument shows that
 $\underline{x}_H (t)$ is  monotonically non-decreasing. Therefore $x_H(t)$ is uniformly bounded and the limits
\begin{equation}
\label{limvalues}
 \bar{x}_H^{\infty}  := \lim_{t \rightarrow + \infty} \bar{x}_H(t) \qquad  \underline{x}_H^{\infty} := \lim_{t \rightarrow + \infty} \underline{x}_H (t), 
\end{equation}
exist finite. 
For future reference, it is convenient to define the convex-valued differential inclusion given below:
\begin{equation}
\label{timeinvariantdinc}
\dot{z} \in F_H (z) := \textrm{co} \left ( \bigcup_{x_F \in K} \{ f_H (z,x_F) \}  \right ).
\end{equation}
Due to compactness of $K$, and Lipschitz continuity of $f_H$, $F_H$ is a Lipschitz continuous set-valued map.
In particular, $x_H(t)$, is also a (bounded) solution of (\ref{timeinvariantdinc}).
Consider next the associated $\omega$-limit set, which, by boundedness of $x_H(t)$, is non-empty and compact:
\begin{equation}
\Omega_H := \left  \{ x \in \mathbb{R}^{|H|}: \exists \, \{  t_{n} \}_{n=1}^{+ \infty}: \lim_{n \rightarrow + \infty} t_n = + \infty \textrm{ and } x = \lim_{n \rightarrow + \infty} x_H (t_n) \right \}.
\end{equation}
Notice that, by definition, for any $z_H \in \Omega_H$ we have $\bar{z}_H = \bar{x}_H^{\infty}$ and $\underline{z}_H =  \underline{x}_H^{\infty}$.
As is well known, $\Omega_H$ is a weakly invariant set for the differential inclusion (\ref{timeinvariantdinc}).
Selecting any element $\tilde{z}_H$ in $\Omega_H$, there exists at least one viable solution $\tilde{z}_H (t)$ of (\ref{timeinvariantdinc}), such that
$\tilde{z}_H (t) \in \Omega_H$, for all $t$.
Notice that, by Lipschitzness of $F_H$, the sets 
\[ M(t) := \{ h \in H: \tilde{z}_h (t) = \bar{x}_H^{\infty} \}, 
\]
and
\[ m(t) :=  \{ h \in H: \tilde{z}_h (t) = \underline{x}_H^{ \infty} \} \]
are monotonically non-increasing with respect to set-inclusion and, trivially, non-empty for all $t \geq 0$.
Hence, there exists some finite $\tau\geq 0$ such that $M(t)=M(\tau)$ and $m(t)=m (\tau)$ for all $t \geq \tau$.
Moreover, for all such values of $t$, we see that:
\begin{equation}
\label{isconstant}
 \dot{\tilde{z}}_h (t) = 0 \qquad \forall \, h \in M( \tau),
\end{equation}
  and similarly 
\begin{equation}
\label{isconstant2}
 \dot{\tilde{z}}_h (t) = 0 \qquad \forall \, h \in m( \tau).
\end{equation} 
To prove asymptotic consensus, we need to show $M(\tau) \cap m(\tau) \neq \emptyset$. 
To this end, we claim that there exists $F_M \subset F$ such that $M( \tau )$ is an $F_M$ controlled siphon, as we argue next  by contradiction. \\

Should this not happen, at least some $h$  would exist in $M(t)$ and $I \subset \mathcal{N}$ such that $I \rightarrow h$ and still $I \cap ( M(t) \cup F ) = \emptyset$. In particular then, $\bar{ \tilde{z}}_I(t):= \max_{i \in I}  \tilde{z}_i (t) < \bar{x}_H^{ \infty}$ and this violates (\ref{isconstant}) by virtue of definition (\ref{jointinter}) as for all $x_F \in K$:
\[  f_h (\tilde{z}_H (t), x_F)  \leq f_h ( \bar{x}_H^{\infty} \textbf{1}_H + (  \bar{\tilde{z}}_I (t) - \bar{x}_H^{\infty} )  e_I, x_F ) \qquad \qquad \qquad \]
\[ \qquad \qquad \qquad \qquad \leq - \rho ( \bar{x}_H^{\infty}-\bar{\tilde{z}}_I (t)  )<0. \]
A similar argument can be used to show that $m(t)$ is an $F_m$ controlled siphon. \\

Consider next any switch pairs $F_M, F_m \subset F $ such that $M( \tau )$ and $m( \tau)$ are, respectively, an $F_M$ and $F_m$ controlled siphon.
Assume, without loss of generality, $F_M$ and $F_m$ minimal with respect to set inclusion (among similar siphons' switches). \\
    
In the following we argue by contradiction considering the case $M( \tau ) \cap m( \tau ) = \emptyset$. 
By structural consensuability this implies $F_M \cap F_m \neq \emptyset$ and we may pick $\bar{f} \in  F_M \cap F_m$.
By minimality of $F_M$ and $F_m$, moreover, taking out $\bar{f}$ from them violates the definition of controlled siphon, viz. there exist $h_M \in M(\tau)$ and $h_m \in m(\tau)$ (distinct from each other), such that
for some joint interactions $I_M \rightarrow h_M$ and $I_m \rightarrow h_m$ we see that
\begin{equation}
\label{fshouldbeabove}
 I_M \cap (  M(t) \cup ( F_M \backslash \{ \bar{f} \} )  ) = \emptyset,
\end{equation}
and, similarly,
\begin{equation}
\label{fshouldbebelow}
 I_m \cap ( m(t) \cup ( F_m \backslash \{ \bar{f} \} )  ) = \emptyset. 
\end{equation}


Since $H$ is a siphon, however, $f_{h_M} (z, x_F) \leq 0$ for all $z \in \Omega_H$ and all $x_F \in K$.  
Moreover condition (\ref{fshouldbeabove}) yields:
\[ x_{\bar{f}} < \bar{x}_H^{\infty} \Rightarrow
f_{h_M} ( \tilde{z}_H(t), x_F ) \leq - \rho ( \bar{x}_H^{\infty} - \max \{ \bar{\tilde{z}}_{I_M}(t), x_{\bar{f}} \} ) < 0.\] 
Similarly, $f_{h_m} (z,x_F) \geq 0$ for all $z \in \Omega_H$ and all $x_F \in K$. In addition,
\[ x_{\bar{f}} > \underline{x}_H^{\infty} \Rightarrow f_{h_m} ( \tilde{z}_H(t), x_F ) \geq  \rho ( \min \{ \underline{\tilde{z}}_{I_m}(t), x_{\bar{f}} \} - \underline{x}_H^{\infty} ) > 0. \]
Notice that, whenever $\underline{x}_H^{\infty} <  \bar{x}_H^{\infty} $, we have  $(- \infty,  \bar{x}_H^{\infty} ) \cup ( \underline{x}_H^{\infty}, + \infty ) = \mathbb{R}$ and therefore,
\[ f_{h_M} ( \tilde{z}_H(t), x_F ) -  f_{h_m} ( \tilde{z}_H(t), x_F )  \qquad \qquad \qquad \qquad \] \[\leq - \rho ( \min \big \{  \bar{x}_H^{\infty} - \max \{ \bar{\tilde{z}}_{I_M}(t),  (\underline{x}_H^{\infty} +  \bar{x}_H^{\infty})/2 \},  \min \{ \underline{\tilde{z}}_{I_m}(t), (\underline{x}_H^{\infty} + 
 \bar{x}_H^{\infty})/2 \} - 
\underline{x}_H^{\infty} \big \} ). \]
As a consequence:
\[  ( e_{h_M} - e_{h_m} )' F_H ( \tilde{z}_H(t), x_F) \qquad \qquad \qquad \]
\[ \leq - \rho ( \min \big \{  \bar{x}_H^{\infty} - \max \{ \bar{\tilde{z}}_{I_M}(t),  (\underline{x}_H^{\infty} +  \bar{x}_H^{\infty})/2 \},  \min \{ \underline{\tilde{z}}_{I_m}(t), (\underline{x}_H^{\infty} + 
 \bar{x}_H^{\infty})/2 \} - 
\underline{x}_H^{\infty} \big \} ) < 0 \]
for all $x_F \in K$.
This, however, contradicts either (\ref{isconstant}) or (\ref{isconstant2}). \\
\end{proof}
Notice that, in the proof of Theorem \ref{mr},  it is crucial that the position communicated by malicious agents to all of its neighbors are consistent.
If not, malicious agents could more easily prevent consensus by sending differentiated signals to individual agents, and a correspondingly stronger notion of structural consensuability would be needed.
\section{Examples and Simulation}
\label{simuex}
We consider next an example with $9$ agents arranged in a $3 \times 3$ grid.
Each agent is denoted by an (ordered) pair of integers in $\{1,2,3\}:=N$. In particular then $\mathcal{N}= N \times N$.
We consider the following interconnection topology.
For all $(i,j) \in N \times N$, we have two joint-agent interactions:
\[ ( N \backslash \{i \} ) \times \{ j \} \rightarrow (i,j) \]
\[ \{ i \} \times ( N \backslash \{j\} ) \rightarrow (i,j). \]
The associated Petri Net is shown in Fig. \ref{petrimess}.
\begin{figure}
\centerline{
\includegraphics[width=9cm]{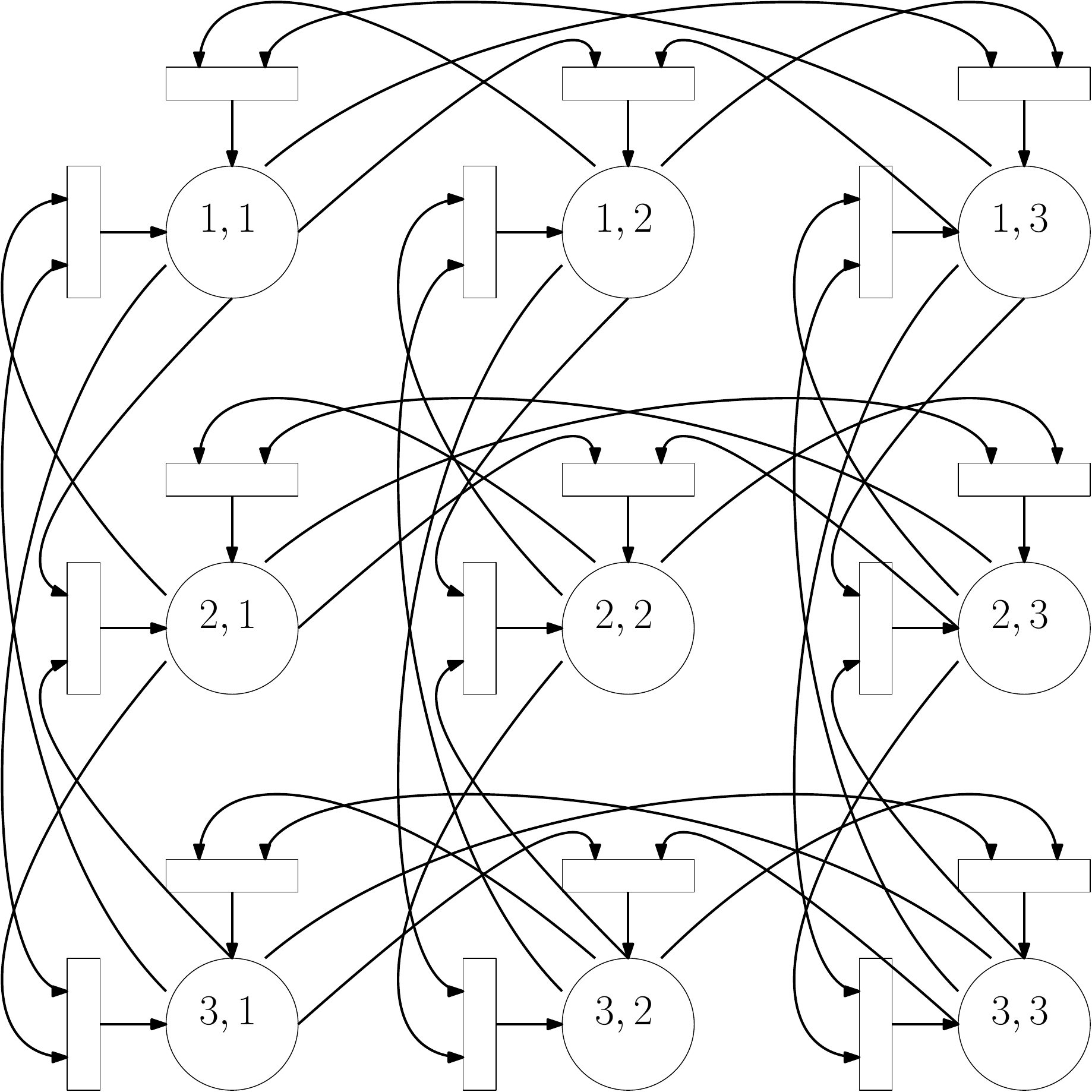}}
\caption{Petri Net associated to joint-agent interactions}
\label{petrimess}
\end{figure}
Notice that, by construction, whenever an agent belongs to a siphon, somebody from the same column and row also needs to be within the siphon.
Let, for a set $\Sigma \subset N \times N$, $\Sigma_i$ denote the elements of $\Sigma$ belonging to $\{i\} \times N$ and, by $\Sigma^j$ the elements of $\Sigma$ in $N \times \{j \}$.
We see that $\Sigma$ (non-empty) is a siphon if and only if $|\Sigma_i| \geq 2$ for all $i \in N$ such that $|\Sigma_i|>0$ and 
$|\Sigma^j| \geq 2$ for all $j \in N$ such that $|\Sigma^j|>0$.
In particular, minimal siphons fulfill the equality rather than the strict inequality and are essentially of two kinds, as shown in Fig.
\ref{basicsi}.
It is straightforward to see that, despite having siphons of cardinality strictly smaller than half of the size of the group ($4 < 9/2$ ), still their layout ensures that two siphons will always share at least one element. Therefore, structural consensuability is fulfilled and asymptotic consensus is guaranteed in the absence of faulty agents for all initial conditions. \\

 \begin{figure}
\centerline{
(a) \includegraphics[width=4.5cm]{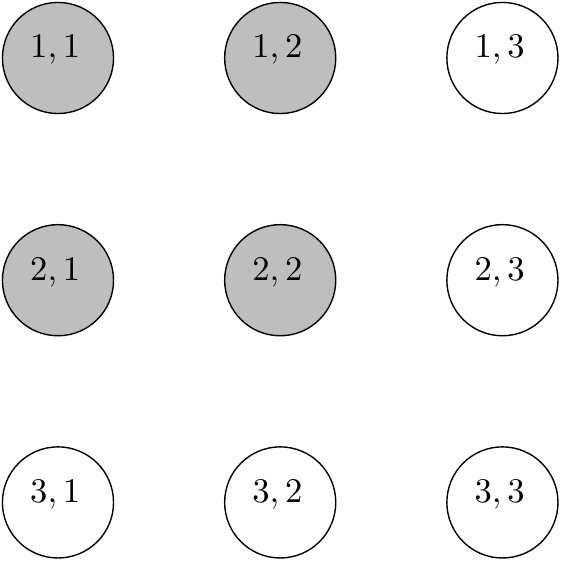} \qquad \qquad
(b) \includegraphics[width=4.5cm]{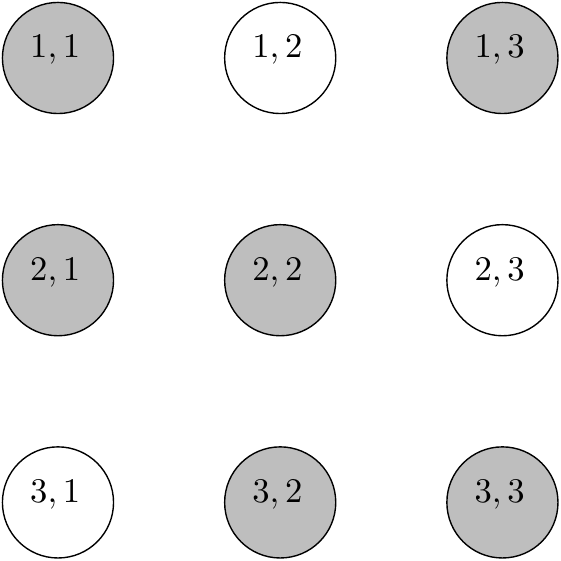}
}
\caption{Siphons of minimal support (gray)}
\label{basicsi}
\end{figure}

Next we investigate the possibility of achieving robust consensus in the presence of a single faulty agent. Due to the simmetry of the network we may choose any agent to be faulty, and the corresponding analysis will apply to any other possible agent after some permutations.
For ease of graphical representation we choose the faulty agent to be $(2,2)$. \\

First of all it is easily seen that $N \times N \backslash \{ (2,2) \}$ is a siphon.
Moreover, any siphon of the full Petri Net that does not contain $(2,2)$ is also a $\emptyset$-controlled siphon when $F= \{ (2,2) \}$.
Next, we look for $(2,2)$-controlled siphons. These are siphons for which agent $(2,2)$ acts as a switch.
It can be seen that $\Sigma$ is a $(2,2)$ controlled siphon if (and only if) $\Sigma \cup \{(2,2)\}$ is a siphon.
This direct implication is true for all Petri Nets, but the converse need not hold in general.
In particular, then, only two types of controlled siphons can be identified (up to permutations), as shown in Fig. \ref{contrsi}.

 \begin{figure}
\centerline{
(a) \includegraphics[width=4.5cm]{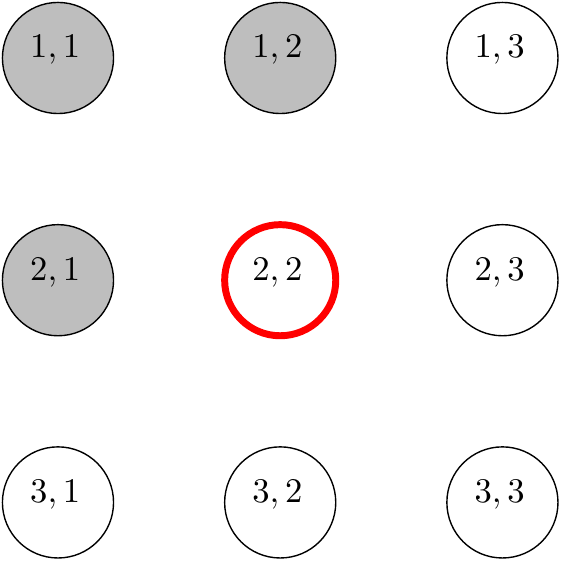} \qquad \qquad
(b) \includegraphics[width=4.5cm]{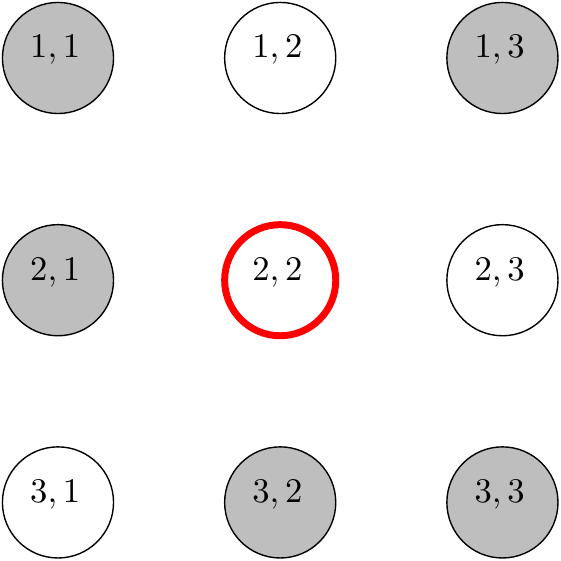}}
\caption{$(2,2)$-controlled siphons of minimal support (gray)}
\label{contrsi}
\end{figure}

Because of this, for any pair of $\emptyset$ or $(2,2)$ controlled siphons $\Sigma_1$, $\Sigma_2$, with associated switches $F_1$, $F_2$ it is true that
\[ (\Sigma_1 \cup F_1) \cap (\Sigma_2 \cup F_2) \neq \emptyset. \]
And, since by definition $\Sigma_i \cap F_i = \emptyset$, then the following is true:
\[ \Sigma_1 \cap \Sigma_2 = \emptyset \Rightarrow F_1 \cap F_2 \neq \emptyset. \]

Hence, robust structural consensuability is fulfilled and one may expect the $8$ healty agents to reach asymptotic consensus despite the exogenous disturbance coming from agent $(2,2)$. As previously remarked, this is still true for all possible selections of a single faulty agent. \\
\newline
Next we explain why this network is not able, in general, to withstand more than a single faulty agent. If two faulty agents occur, either in the same row or column, then the set of healty agents will exhibit a row or a column with a single agent. This implies that the set of healthy agents is not a siphon. Hence robust structural consensuability is not fulfilled. Indeed, the two agents have the ability to influence the agent within the same column (or row) and disrupt its ability to reach asymptotic consensus. \\

Consider next the case of two faulty agents that are not in the same row or column. For instance agents $(2,2)$ and $(3,3)$. Let 
$F = \{ (2,2), (3,3) \}$. As we already characterized s$\emptyset$ controlled siphons and siphons controlled by switch of cardinality one, we need only look for $F$-controlled siphons.
Any set $\Sigma$ such $\Sigma \cup F$ is a siphon is also an $F$-controlled siphon. In addition, the network exhibits two types of $F$-controlled siphons that do not fulfill such condition.
These are shown in Fig. \ref{f1f2si}. \\
 \begin{figure}
\centerline{
(a) \includegraphics[width=4.5cm]{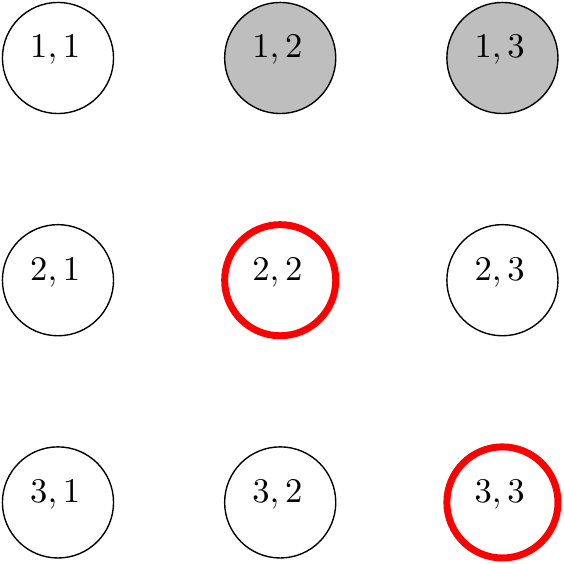} \qquad \qquad
(b) \includegraphics[width=4.5cm]{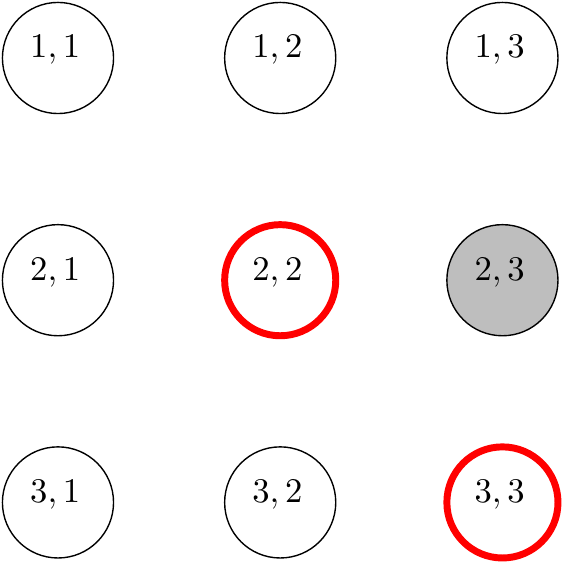}}
\caption{$\{ (2,2), (3,3) \}$-controlled siphons of minimal support (gray)}
\label{f1f2si}
\end{figure}
Notice that, $\{ (2,3) \}$ is an $F$-controlled siphon. On the other hand, the set $\{(1,1),(1,2),(3,1),(3,2)\}$ is an $\emptyset$-controlled siphon.
These two controlled siphons and their associated switches have both empty intersection. Hence, robust structural consensuability does not hold for this choice of faulty agents.
Indeed agents in $F$ have the ability to prevent consensus between the agents in the siphons described above.
\begin{figure}
\centerline{
\includegraphics[width=10cm]{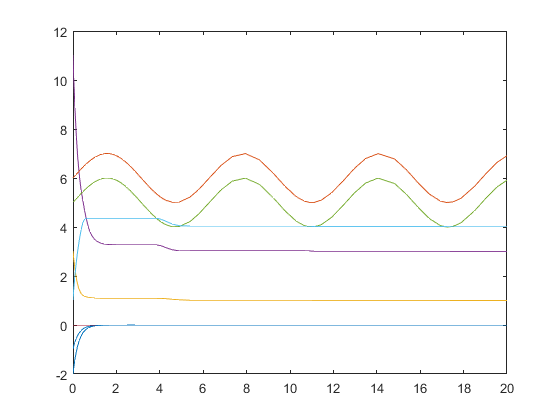}}
\caption{Malicious agents $(2,2)$ and $(3,3)$ disrupt consensus.}
\label{notrob}
\end{figure}
For instance, one may take initial conditions $x_{1,1}(0)=x_{1,2} (0) = x_{3,1} (0) = x_{3,2}(0)= 1$. For a $\emptyset$-controlled siphon this results in solutions
fulfilling $x_{1,1}(t)=x_{1,2} (t) = x_{3,1} (t) = x_{3,2}(t)= 1$ for all $t$. At the same time, one may let the malicious agents fulfill $x_{3,3}(t)=x_{2,2} (t)=0$ so that,
any solution with $x_{2,3}(0)=0$ would result in $x_{2,3}(t)=0$ identically, thus preventing consensus. A similar issue arises, letting agents $x_{1,1},x_{1,2},x_{3,1},x_{3,2}$ be initialized with negative values,
while $x_{2,2}$ and $x_{3,3}$ oscillate at some higher values as shown in Fig. \ref{notrob}.

Due to simmetry of the considered network, it follows that selection of any two faulty agents result in the conditions for robust consensuability to be violated. 

We emphasize that the considered network is not robust with respect to Byzantine malicious agents. In particular, in the simulation we show the result of  agent $(2,2)$ broadcasting higher values to agents $(1,2)$, $(2,1)$ than the one broadcast to agents $(3,2)$, $(2,3)$. The Byzantine agent is initialized with $x_{2,2} (0)=0$ and does not change his position, while it sends the value $+2$ and $-2$ to its neighbors, thus disrupting consensus as shown in Fig. \ref{byz}. It is worth pointing out that if the agent had consistently sent the same information to all his neighbours robust consensus would have been achieved.
\begin{figure}
\centerline{
\includegraphics[width=10cm]{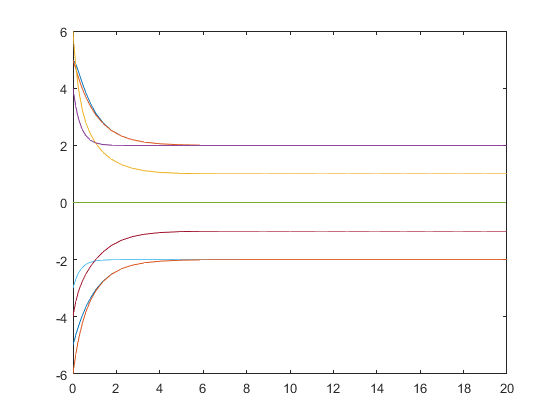}}
\caption{Byzantine agent $(2,2)$ disrupts consensus.}
\label{byz}
\end{figure}

\section{Comparison with ARC-P protocols}
\label{arcpsection} 
Adversarily Robust Consensus was first introduced by LeBlanc and coworkers in \cite{leblanc1}. This is proposed, initially, for all-to-all networks and later extended in \cite{leblanc3} to networks with more general topologies.
We start this Section by highlighting how the all-to-all topology considered in \cite{leblanc1} can be seen as a specific type of symmetric joint-agent interaction. Similar considerations apply when the set of neighbours of each agent is a proper subset of $\{1,2, \ldots, n \}$ but, for the sake of simplicity, this is not illustrated in detail.
Let $\bar{\sigma}_k(x)$ denote the $k$-th largest entry in $x$ and, similarly, $\underline{\sigma}_k (x)$ the $k$-th smallest entry in $x$.
We see that,
\[ \bar{\sigma}_k (x) = \max_{J \subset \mathcal{N}: |J|= k  } \;  \min_{k \in J} x_k \]
\[ \underline{ \sigma }_k (x) =  \min_{J \subset \mathcal{N}: |J|= k }  \;  \max_{k \in J} x_k. \]
 Moreover, $\bar{ \sigma }_k (x) = \underline{ \sigma }_{n+1-k} (x)$, and therefore for any integer $F$ with $n-F \geq F+1$ we see
\[ 
\sum_{k= F+1}^{n-F}  \bar{\sigma}_k (x) =  \sum_{k= F+1}^{n-F} \underline{\sigma}_k (x). \]
Consider next the protocol described by the following set of equations:
\begin{equation}
\label{forcomparison}
 \dot{x}_i = -x_i + \frac{ \sum_{k= F+1}^{n-F}  \bar{\sigma}_k (x)}{n-2F}. 
\end{equation}
This is, essentially, a continuous-time version of the algorithm proposed in \cite{leblanc1}, where each agent is directed towards the average of the $n-2F$ agents' opinions of intermediate value (as achieved in \cite{leblanc1} by using the sorting and reducing maps).
It is easy to see that (\ref{forcomparison}) is a monotone cooperative network, moreover, we claim that for all $J$ of cardinality $F+1$ and any agent $i$ it holds $J \rightarrow i$.
To this end, let $J$ be a subset of cardinality $F+1$ and let $x_J>x_i$ be the common value associated with agents in $J$. All agents not in $J$, including agent $i$ have value $x_i$, instead. 
Clearly $\bar{\sigma}_k(x) = x_J$ for all $k=1 \ldots F+1$ and $\bar{ \sigma}_k (x) = x_i$ for all $k= F+2 \ldots n$.
In particular, then:
\[ \dot{x}_i = - x_i + \frac{ x_J + (n - 2 F - 1) x_i}{n - 2F} = \frac{ x_J - x_i }{ n -2F }, \] 
which proves a joint influence of agents in $J$ towards $i$ from above. Similar results hold when $x_J < x_i$. Moreover,
$J \rightarrow i$ is a minimal influence. In fact, any proper subset of $J$ consists of at most $F$ elements and therefore, assuming their value is $x_J$ while $x_i$ is the value of other agents, a simple  computation shows that
\[ \dot{x}_i = - x_i + \frac{ (n - 2F) x_i }{n - 2F} = 0. \]
thus ruling out the possibility of joint-influences for sets of agents of cardinality $F$ or lower.
Similar computations can be carried out when the network is not of all-to-all type, and each agent has a specific set of neighbours that get sorted, reduced and averaged upon.
 
\section{Conclusions}
This paper has explored tight necessary and sufficient conditions for continuous-time Adversary Consensus Protocols of networks with joint-agent interactions of arbitrary topology. This captures, as a particular case, the notion of ARC consensus studied in the discrete-time case by Leblanc and co-workers, using \emph{sorting} and \emph{selection} maps.
Consensus is achieved in the face of agents that behave as arbitrary bounded disturbances, and are only constrained to broadcasting the same information to all of the neighbours they have an influence upon.
In this respect, the problem of \emph{Byzantine} consensus, where agents may maliciously or unintenionally send different information to distinct neighbours is an interesting open question for further research.
Conditions are formulated in the language of Petri Nets, in particular making use of the notion of \emph{controlled siphon}, in which faulty agents play the role of a `switch' capable of disabling some influences by suitably positioning itself above or below the value of agents within the same joint-agent interaction.
An example is presented to illustrate the applicability of the considered results. This does not fall within the class of networks considered in \cite{leblanc3} since each agent is only has two distint group of neighbours (vertical and horizontal ones in the picture) which are treated separately when cross-validating information in joint-agent interactions). In particular, the equations considered can never be achieved by means of sorting and selection functions.

\end{document}